\documentclass[onecolumn, twoside, 11pt, draftclsnofoot]{IEEEtran}
\usepackage[T1]{fontenc}
\usepackage{stfloats, balance, enumerate, cite, xcolor, amsmath, mathtools, cuted, xcolor, mathrsfs, mathtools, amssymb, amsthm, bigints, multirow, mleftright}

\newtheorem{theorem}{Theorem}
\ifCLASSINFOpdf
   \usepackage[pdftex]{graphicx}
\else
\fi
\begin{document}

\title{Performance Analysis of NOMA-based Cooperative Relaying in $\alpha$-$\mu$ Fading Channels}

\author{\IEEEauthorblockN{Vaibhav Kumar,
Barry Cardiff, and Mark F. Flanagan} \\
\IEEEauthorblockA{School of Electrical and Electronic Engineering \\ 
University College Dublin, Belfield, Dublin 4, Ireland\\
Email: vaibhav.kumar@ucdconnect.ie, barry.cardiff@ucd.ie, mark.flanagan@ieee.org}}

\maketitle
\let\bs\boldsymbol

\begin{abstract}
Non-orthogonal multiple access (NOMA) is widely recognized as a potential multiple access technology for efficient radio spectrum utilization in the fifth-generation (5G) wireless communications standard. In this paper, we study the average achievable rate and outage probability of a cooperative relaying system (CRS) based on NOMA (CRS-NOMA) over wireless links governed by the $\alpha$-$\mu$ generalized fading model; here $\alpha$ and $\mu$ designate the \emph{nonlinearity} and \emph{clustering} parameters, respectively, of each link. The average achievable rate is represented in closed-form using Meijer's G-function and the extended generalized bivariate Fox's H-function (EGBFHF), and the outage probability is represented using the lower incomplete Gamma function. Our results confirm that the CRS-NOMA outperforms the CRS with conventional orthogonal multiple access (CRS-OMA) in terms of spectral efficiency at high transmit signal-to-noise ratio (SNR). It is also evident from our results that with an increase in the value of the nonlinearity/clustering parameter, the SNR at which the CRS-NOMA outperforms its OMA based counterpart becomes higher. Furthermore, the asymptotic analysis of the outage probability reveals the dependency of the diversity order of each symbol in the CRS-NOMA system on the $\alpha$ and $\mu$ parameters of the fading links.
\end{abstract}


\IEEEpeerreviewmaketitle

\section{Introduction}
NOMA has gained widespread popularity as a potential candidate for multiple access technology in the 5G wireless standard and beyond, which is capable of satisfying the need for low-latency and high-reliability communications for heterogeneous data traffic~\cite{Bhargava}. Many variants of NOMA have been suggested in the literature, e.g., power-domain NOMA, sparse code multiple access~(SCMA), interleave division multiple access~(IDMA), low-density spreading (LDS) and pattern division multiple access~(PDMA)~\cite{NOMA_book}. Among all of these variants, power-domain NOMA (which we will refer to simply as `NOMA' throughout this paper) has gained widespread popularity because of its implementation-friendly transceiver architecture. It can accommodate several users within the same orthogonal resource block (time, frequency and/or spreading code) via multiplexing them in the power domain at the transmitter and using successive interference cancellation (SIC) at the receiver to remove the messages intended for other users. In the case of NOMA, users with poor channel conditions have a larger share of the transmission power, unlike the conventional OMA where more power is allocated to users with stronger channels (also known as the water-filling strategy)~\cite{NOMA_App}.

A very promising application of NOMA for power-domain multiplexed transmission using a cooperative relaying system (CRS-NOMA) was proposed in~\cite{CRS_NOMA}, where the source was able to deliver two different symbols to the destination in two time slots with the help of a relay. It was shown in~\cite{CRS_NOMA} that CRS-NOMA outperforms the conventional OMA based decode-and-forward cooperative relaying system in terms of average achievable rate in the case of Rayleigh fading channels. The performance superiority of CRS-NOMA in Rician fading was shown in~\cite{CRS_NOMA_Rician}. However, while Rayleigh and Rician fading models are analytically convenient wireless channel models, situations are easily encountered in practice for which these distributions do not adequately fit experimental data~\cite{MDY}.

In this paper, we present the performance analysis of CRS-NOMA in the presence of $\alpha$-$\mu$ fading, which is a generalized, flexible and mathematically tractable distribution to model the effect of small-scale channel fading and includes Gamma, Erlang, Chi-squared, Rayleigh, Nakagami-$m$, exponential, Weibull and one-sided Gaussian distributions as special cases~\cite{MDY}. It can also (approximately) represent some large-scale fading (such as log-normal fading), as described in~\cite{LogNormal}. Moreover, the very high degree of flexibility afforded by the $\alpha$-$\mu$ distribution renders it better suited to adjust to fit field data.

The main contributions of this paper include the derivation of  closed-form expressions for the average achievable rate and outage probability of the CRS-NOMA in the presence of $\alpha$-$\mu$ fading on each link. In order to have a better insight into the asymptotic system performance, we also present a diversity analysis of the CRS-NOMA, and in particular we show how the diversity order depends on the $\alpha$ and $\mu$ parameters of each of the fading links.
\section{System Model}
Consider the cooperative relaying system shown in Fig.~\ref{SysMod}, which consists of a source $S$, a relay $R$ and a destination $D$. All nodes are equipped with a single antenna and are assumed to be operating in half-duplex mode. All wireless links are assumed to be $\alpha$-$\mu$ distributed with the non-linearity parameters between $S$-$R$, $S$-$D$ and $R$-$D$ links denoted by $\alpha_{sr}$, $\alpha_{sd}$ and $\alpha_{rd}$, respectively, the corresponding clustering parameters denoted by $\mu_{sr}$, $\mu_{sd}$ and $\mu_{rd}$, respectively, and the $\alpha$-root-mean values of the channel fading coefficients denoted by $\Omega_{sr}, \Omega_{sd}$ and $\Omega_{rd}$, respectively. Furthermore, it is assumed that the $S$-$D$ link is on average weaker that the $S$-$R$ link, i.e., $\Omega_{sd}^{\alpha_{sd}} < \Omega_{sr}^{\alpha_{sr}}$.

In the CRS-NOMA, the source broadcasts $\sqrt{a_1 P}s_1 + \sqrt{a_2 P}s_2$ to both relay and destination in the first time slot, where $s_1$ and $s_2$ are data-bearing constellation symbols, multiplexed in the power domain ($\mathbb E[|x_i|^2] = 1$ for $i = 1, 2$, where $\mathbb E[\cdot]$ denotes the expectation operator), $P$ is the power transmitted from the source, and $a_1$ and $a_2$ are power allocation coefficients (with $a_1 + a_2 = 1$ and $a_1 > a_2$). After reception, the destination decodes symbol $s_1$ by treating the interference from symbol $s_2$ as additional noise, while the relay first decodes symbol $s_1$ (treating the interference from symbol $s_2$ as additional noise) and then applies SIC to decode symbol $s_2$. In the second time slot, the relay transmits its estimate of symbol $s_2$, denoted by $\hat s_2$, to the destination with power $P$. In this manner, two different symbols are delivered to the destination in two time slots.

In contrast to this, in the CRS-OMA, the source broadcasts symbol $s_1$ with power $P$ to both relay and destination in the first time slot and the relay (after decoding) forwards its estimate of symbol $s_1$, denoted by $\hat s_1$, to the destination in the second time slot. In this fashion, only a single symbol is delivered to the destination in two time slots.
\begin{figure}[t]
\centering
\includegraphics[scale = 0.8]{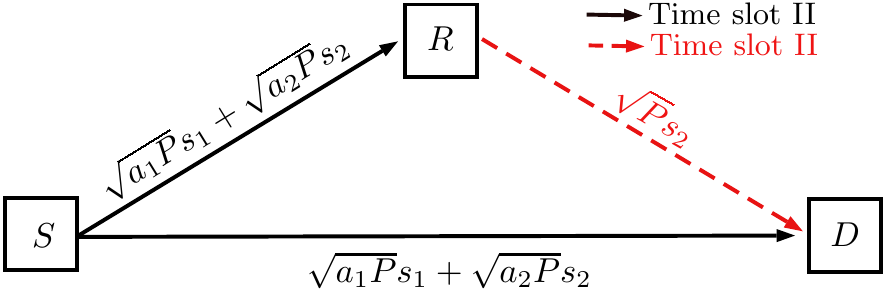}
\caption{System model for NOMA-based cooperative relaying.}
\label{SysMod}
\end{figure}

Denoting the channel fading coefficient between the $S$-$R$, $S$-$D$ and $R$-$D$ links by $h_{sr}$, $h_{sd}$ and $h_{rd}$, respectively, the probability density function (PDF) of $|h_{\tau}|, \tau \in \{sr, sd, rd\}$ can be expressed as~\cite[Eqn.~(1)]{MDY}
\begin{align}
	f_{|h_\tau|}(x) = \dfrac{\alpha_\tau \mu_\tau^{\mu_\tau}x^{\alpha_\tau \mu_\tau - 1}}{\Omega_\tau^{\alpha_\tau \mu_\tau} \Gamma(\mu_\tau)} \exp \left( -\mu_\tau \dfrac{x^{\alpha_\tau}}{\Omega_{\tau}^{\alpha_\tau}}\right), \notag 
\end{align}
where $\Gamma(\cdot)$ is the Gamma function. The cumulative distribution function (CDF) of the channel gain $\lambda_\tau \triangleq |h_\tau|^2$ can therefore be given by
\begin{align}
	F_{\lambda_\tau}(x) = \int_0^{\sqrt x} f_{|h_\tau|}(t)dt = \dfrac{1}{\Gamma(\mu_\tau)}\gamma\left( \mu_\tau, \dfrac{\mu_\tau x^{0.5 \alpha_\tau}}{\Omega_\tau^{\alpha_\tau}}\right), \!\! \label{F_lambda_tau}
\end{align}
where $\gamma (\cdot, \cdot)$ is the lower incomplete Gamma function and the integration above is solved using~\cite[Eqn.~($3.381$-$8^*$),~p.~346]{Grad}. The PDF of $\lambda_{\tau}$ can be obtained by differentiating the equation above w.r.t. $x$, yielding
\begin{align}
	f_{\lambda_\tau}(x) = \dfrac{\alpha_\tau \mu_\tau^{\mu_\tau} x^{0.5 \alpha_\tau \mu_\tau - 1}}{2 \Omega_\tau^{\alpha_\tau \mu_\tau} \Gamma(\mu_\tau)}\exp \left( -\mu_\tau \dfrac{x^{0.5\alpha_\tau}}{\Omega_\tau^{\alpha_\tau}}\right). \label{f_lambda_tau}
\end{align}

\section{Performance Analysis}
In this section we present a comprehensive analysis of the average achievable rate, outage probability and diversity order for the CRS-NOMA in $\alpha$-$\mu$ fading channels.
\subsection{Average achievable rate}
The signals received at the relay and the destination in the first time slot are given by
\begin{align}
	y_{sr} = & h_{sr}\left( \sqrt{a_1 P}s_1 + \sqrt{a_2 P} s_2\right) + n_{sr}, \notag \\
	y_{sd} = & h_{sd} \left( \sqrt{a_1 P}s_1 + \sqrt{a_2 P} s_2\right) + n_{sd}, \notag 
\end{align}
where $n_{sr}$ and $n_{sd}$ denote complex additive white Gaussian noise (AWGN) with zero mean and variance $\sigma^2$ at the relay and the destination, respectively. The received instantaneous signal-to-interference-plus-noise ratio (SINR) at the relay for decoding symbol $s_1$ and the received instantaneous SNR for decoding symbol $s_2$ (assuming symbol $s_1$ is decoded correctly) is given by
\begin{align}
	\gamma_{sr}^{(1)} = \dfrac{a_1 \rho \lambda_{sr}}{a_2 \rho \lambda_{sr} + 1}, \quad \gamma_{sr}^{(2)} = a_2 \rho \lambda_{sr}, \notag
\end{align}
where $\rho = P/\sigma^2$ is the transmit SNR. Similarly, the received instantaneous SINR for decoding symbol $s_1$ at the destination is given by
\begin{align}
	\gamma_{sd} = \dfrac{a_1 \rho \lambda_{sd}}{a_2 \rho \lambda_{sd} + 1}. \notag
\end{align}
In the second time slot, the relay forwards it estimate of $s_2$, denoted by $\hat s_2$, to the destination with power $P$. The signal received at the destination is given by
\begin{align}
	y_{rd} = \sqrt{P} \hat s_2 + n_{rd}, \notag
\end{align}
where $n_{rd}$ denotes the complex AWGN at the destination with zero mean and variance $\sigma^2$. The received instantaneous SNR at the destination for decoding symbol $s_2$ is given by
\begin{align}
	\gamma_{rd} = \rho \lambda_{rd}. \notag
\end{align}
The instantaneous achievable rate for symbol $s_1$ is therefore given by 
\begin{align}
	C_{s_1} = & \min \left\{0.5 \log_2\left(1 + \gamma_{sr}^{(1)}\right), 0.5 \log_2(1 + \gamma_{sd})\right\} \notag \\
			= & \, 0.5 \log_2\left( 1 + \dfrac{a_1 \rho X}{a_2 \rho X + 1}\right) = 0.5 \log_2(1 + \rho X) - 0.5 \log_2(1 + a_2 \rho X), \notag
\end{align}
where $X \triangleq \min\left\{\lambda_{sr}, \lambda_{sd}\right\}$. Similarly, the instantaneous achievable rate for symbol $s_2$ is given by
\begin{align}
	C_{s_2} = & \min \left\{ 0.5 \log_2 \left( 1 + \gamma_{sr}^{(2)}\right), 0.5 \log_2(1 + \gamma_{rd})\right\} = 0.5 \log_2(1 + \rho Y), \notag 
\end{align}
where $Y \triangleq \min\{a_2 \lambda_{sr}, \lambda_{rd}\}$.

The average achievable rate for symbol $s_1$ is therefore given by
\begin{align}
	& \bar C_{s_1} = \mathbb E_{X}\left[ C_{s_1}\right] = \dfrac{0.5}{\ln(2)} \! \left[\int_{0}^{\infty} \!\! \ln(1 + \rho x)f_X(x)dx \! - \! \!\!\int_{0}^{\infty}\! \! \ln(1 \!+ \!a_2 \rho x)f_X(x)dx \right] = \dfrac{0.5}{\ln(2)} \left[ \mathbb I_1 - \mathbb I_2\right], \label{C_s1_int}
\end{align}
where $\mathbb E_{\mathcal Z}[\cdot]$ denotes expectation w.r.t. the random variable $\mathcal Z$ and $f_{\mathcal Z}(\cdot)$ denotes the PDF of the random variable $\mathcal Z$.
\begin{theorem}
	A closed-form expression for the average achievable rate of symbol $s_1$ in CRS-NOMA is given by 
	\begin{align}
		\bar C_{s_1} \! = \! \dfrac{0.5}{\ln(2)}[\underbrace{(I_1 - I_2 + I_3 - I_4)}_{\mathbb I_1} - \underbrace{(I_5 - I_6 + I_7 - I_8)}_{\mathbb I_2}], \label{C_s1_theorem}
	\end{align}
where,
\begin{align}
	I_1 = \dfrac{\mu_{sr}^{\mu_{sr}}}{\sqrt{2} (2\pi)^{\alpha_{sr} - 0.5} \Omega_{sr}^{\alpha_{sr}\mu_{sr}} \Gamma(\mu_{sr}) \rho^{0.5 \alpha_{sr}\mu_{sr}}} G_{2\alpha_{sr}, 2 + 2 \alpha_{sr}}^{2 + 2 \alpha_{sr}, \alpha_{sr}} \!\!\left(\frac{\mu_{sr}^2}{4 \Omega_{sr}^{2\alpha_{sr}} \rho^{\alpha_{sr}}} \left\vert \begin{smallmatrix} \chi_{sr},  \,\,\Delta\left(\alpha_{sr}, 1- 0.5 \alpha_{sr}\mu_{sr}\right) \\[0.6em] \Delta(2, 0), \quad \chi_{sr}, \quad  \chi_{sr}\end{smallmatrix}\right.\!\right), \label{I1}
\end{align}
\begin{align}
	I_2 = \dfrac{H_{1, 0:2, 2:1, 2}^{1, 0:1, 2:1, 1} \left[ \left. \begin{smallmatrix} \left(1-\mu_{sr}:\frac{2}{\alpha_{sr}}, \frac{\alpha_{sd}}{\alpha_{sr}}\right) \\[0.6em] \text{--} \end{smallmatrix}\!\right\vert \left. \begin{smallmatrix} (1, 1), & (1, 1) \\[0.6em] (1, 1), & (0, 1)\end{smallmatrix}\right\vert  \left. \begin{smallmatrix} (1, 1) &\\[0.6em] (\mu_{sd}, 1),& (0, 1)\end{smallmatrix} \right\vert \rho \left( \frac{\Omega_{sr}}{\mu_{sr}^{1/\alpha_{sr}}}\right)^{\!\!2}, \frac{\mu_{sd} \Omega_{sr}^{\alpha_{sd}}}{\mu_{sr}^{\alpha_{sd}/\alpha_{sr}} \Omega_{sd}^{\alpha_{sd}}} \right]}{\Gamma(\mu_{sr}) \Gamma(\mu_{sd})} , \label{I2}
\end{align}
\begin{align}
	I_3 = \dfrac{\mu_{sd}^{\mu_{sd}}}{\sqrt{2} (2\pi)^{\alpha_{sd} - 0.5} \Omega_{sd}^{\alpha_{sd}\mu_{sd}} \Gamma(\mu_{sd}) \rho^{0.5 \alpha_{sd}\mu_{sd}}} G_{2\alpha_{sd}, 2 + 2 \alpha_{sd}}^{2 + 2 \alpha_{sd}, \alpha_{sd}} \left(\frac{\mu_{sd}^2}{4 \Omega_{sd}^{2\alpha_{sd}} \rho^{\alpha_{sd}}} \left\vert \begin{smallmatrix} \chi_{sd},  \,\,\Delta\left(\alpha_{sd}, 1- 0.5 \alpha_{sd}\mu_{sd}\right) \\[0.6em] \Delta(2, 0), \quad \chi_{sd}, \quad  \chi_{sd}\end{smallmatrix}\right. \right) , \label{I3}
\end{align}
\begin{align}
	I_4 = \dfrac{H_{1, 0:2, 2:1, 2}^{1, 0:1, 2:1, 1} \left[ \left. \begin{smallmatrix} \left(1-\mu_{sd}:\frac{2}{\alpha_{sd}}, \frac{\alpha_{sr}}{\alpha_{sd}}\right) \\[0.6em] \text{--} \end{smallmatrix}\!\right\vert \left. \begin{smallmatrix} (1, 1), & (1, 1) \\[0.6em] (1, 1), & (0, 1)\end{smallmatrix}\right\vert  \left. \begin{smallmatrix} (1, 1) &\\[0.6em] (\mu_{sr}, 1),& (0, 1)\end{smallmatrix}\right\vert \rho \left( \frac{\Omega_{sd}}{\mu_{sd}^{1/\alpha_{sd}}} \right)^{2}, \frac{\mu_{sr} \Omega_{sd}^{\alpha_{sr}}}{\mu_{sd}^{\alpha_{sr}/\alpha_{sd}} \Omega_{sr}^{\alpha_{sr}}} \right]}{\Gamma(\mu_{sd}) \Gamma(\mu_{sr})}, \label{I4}
\end{align}
$G_{p, q}^{m, n}(\cdot)$ is Meijer's G-function, $\chi_{\tau} \triangleq \Delta(\alpha_{\tau}, -0.5 \alpha_{\tau}\mu_{\tau})$ where $\Delta(x, y) \triangleq \tfrac{y}{x}, \tfrac{y + 1}{x}, \ldots, \tfrac{y + x - 1}{x}$, and $H_{p_1, q_1:p_2, q_2:p_3, q_3}^{m_1, n_1:m_2, n_2:m_3, n_3}[\cdot]$ is the extended generalized bivariate Fox's H-function (EGBFHF)\footnote{A \textsc{Matlab} implementation of the EGBFHF is given in~\cite{MatlabImplement} and a \textsc{Mathematica} implementation is given in~\cite{Secrecy}.} as defined in~\cite[Eqn.~(2.57)]{HFunctionBook}. Expressions for $I_5, I_6, I_7$ and $I_8$ can be obtained by replacing $\rho$ by $a_2 \rho$ in~\eqref{I1}-\eqref{I4}, respectively.
\end{theorem}
\begin{proof}
See Appendix~A.
\end{proof}

Similarly, the average achievable rate for symbol $s_2$ is given by 
\begin{align}
	\bar C_{s_2} = \mathbb E_Y[C_{s_2}] = \dfrac{0.5}{\ln(2)} \int_{0}^{\infty} \ln(1 + \rho y)f_Y(y)dy. \label{C_s2_int}
\end{align}
\begin{theorem}
	A closed-form expression for the average achievable rate of symbol $s_2$ in CRS-NOMA is given by 
	\begin{align}
		\bar C_{s_2} = \dfrac{0.5}{\ln(2)}[I_9 - I_{10} + I_{11} - I_{12}], \label{C_s2_theorem}
	\end{align}
	where
\begin{align}
	I_9 = \dfrac{\mu_{sr}^{\mu_{sr}} G_{2\alpha_{sr}, 2 + 2 \alpha_{sr}}^{2 + 2 \alpha_{sr}, \alpha_{sr}} \!\!\left(\frac{\mu_{sr}^2}{4 a_2^{\alpha_{sr}} \Omega_{sr}^{2\alpha_{sr}} \rho^{\alpha_{sr}}} \left\vert \begin{smallmatrix} \chi_{sr},  \,\,\Delta\left(\alpha_{sr}, 1- 0.5 \alpha_{sr}\mu_{sr}\right) \\[0.6em] \Delta(2, 0), \quad \chi_{sr}, \quad  \chi_{sr}\end{smallmatrix}\right.\!\right)}{\sqrt{2} (2\pi)^{\alpha_{sr} - 0.5} a_2^{0.5\alpha_{sr}\mu_{sr}}\Omega_{sr}^{\alpha_{sr}\mu_{sr}} \Gamma(\mu_{sr}) \rho^{0.5 \alpha_{sr}\mu_{sr}}}, \label{I9}
\end{align}
\begin{align}
	I_{10} = \dfrac{H_{1, 0:2, 2:1, 2}^{1, 0:1, 2:1, 1} \left[ \left. \begin{smallmatrix} \left(1-\mu_{sr}:\frac{2}{\alpha_{sr}}, \frac{\alpha_{rd}}{\alpha_{sr}}\right) \\[0.6em] \text{--} \end{smallmatrix}\right\vert  \left. \begin{smallmatrix} (1, 1), & (1, 1) \\[0.6em] (1, 1), & (0, 1)\end{smallmatrix}\right\vert \left. \begin{smallmatrix} (1, 1) &\\[0.6em] (\mu_{rd}, 1),& (0, 1)\end{smallmatrix} \right\vert \rho \left( \frac{\Omega_{sr}}{\mu_{sr}^{1/\alpha_{sr}}}\!\!\right)^{2}, \dfrac{\mu_{rd} a_2^{0.5 \alpha_{rd}}\Omega_{sr}^{\alpha_{rd}}}{\mu_{sr}^{\alpha_{rd}/\alpha_{sr}} \Omega_{rd}^{\alpha_{rd}}} \right]}{\Gamma(\mu_{sr}) \Gamma(\mu_{rd})}, \label{I10}
\end{align}
\begin{align}
	I_{11} = \dfrac{\mu_{rd}^{\mu_{rd}} G_{2\alpha_{rd}, 2 + 2 \alpha_{rd}}^{2 + 2 \alpha_{rd}, \alpha_{rd}} \!\!\left(\frac{\mu_{rd}^2}{4 \Omega_{rd}^{2\alpha_{rd}} \rho^{\alpha_{rd}}} \left\vert \begin{smallmatrix} \chi_{rd},  \,\,\Delta\left(\alpha_{rd}, 1- 0.5 \alpha_{rd}\mu_{rd}\right) \\[0.6em] \Delta(2, 0), \quad \chi_{rd}, \quad  \chi_{rd}\end{smallmatrix}\right. \right)}{\sqrt{2} (2\pi)^{\alpha_{rd} - 0.5} \Omega_{rd}^{\alpha_{rd}\mu_{rd}} \Gamma(\mu_{rd}) \rho^{0.5 \alpha_{rd}\mu_{rd}}}, \label{I11}
\end{align}
\begin{equation}
	I_{12} = \dfrac{H_{1, 0:2, 2:1, 2}^{1, 0:1, 2:1, 1} \left[ \left. \begin{smallmatrix} \left(1-\mu_{rd}:\frac{2}{\alpha_{rd}}, \frac{\alpha_{sr}}{\alpha_{rd}}\right) \\[0.6em] \text{--} \end{smallmatrix} \right\vert  \left. \begin{smallmatrix} (1, 1), & (1, 1) \\[0.6em] (1, 1), & (0, 1)\end{smallmatrix}\right\vert \left. \begin{smallmatrix} (1, 1) &\\[0.6em] (\mu_{sr}, 1),& (0, 1)\end{smallmatrix}\!\right\vert \rho \left( \frac{\Omega_{rd}}{\mu_{rd}^{1/\alpha_{rd}}} \right)^{2}, \frac{\mu_{sr} \Omega_{rd}^{\alpha_{sr}}  a_2^{-0.5\alpha_{sr}}}{\mu_{rd}^{\alpha_{sr}/\alpha_{rd}}\Omega_{sr}^{\alpha_{sr}}} \right]}{\Gamma(\mu_{rd}) \Gamma(\mu_{sr})}. \label{I12}
\end{equation} 
\end{theorem}
\begin{proof}
	See Appendix~B.
\end{proof}

Using~\eqref{C_s1_theorem} and~\eqref{C_s2_theorem}, a closed-form expression for the average achievable rate for CRS-NOMA in the presence of $\alpha$-$\mu$ fading is given by
\begin{align}
	\bar C = \bar C_{s_1} + \bar C_{s_2}. \label{C_sum}
\end{align}

In contrast to this, for CRS-OMA, the signals received at the relay and the destination in the first time slot are given by
\begin{align}
	y_{sr, \mathrm{OMA}} = & h_{sr} \sqrt{P}s_1 + n_{sr}, \notag \\
	y_{sd, \mathrm{OMA}} = & h_{sd} \sqrt{P}s_1 + n_{sd}. \notag
\end{align}
The relay forwards its estimate of symbol $s_1$, denoted by $\hat s_1$, in the next time slot. The signal received at the destination in the second time slot is given by 
\begin{align}
	y_{rd, \mathrm{OMA}} = & h_{rd} \sqrt{P}\hat s_1 + n_{rd}. \notag
\end{align}
Therefore, the average achievable rate for the CRS-OMA is given by 
\begin{align}
	\bar C_{\mathrm{OMA}} = 0.5 \mathbb E_Z[\log_2 (1 + \rho Z)], \label{C_OMA}
\end{align}
where $Z \triangleq \min\{\lambda_{sr}, \lambda_{sd} + \lambda_{rd}\}$. Since the focus of this paper is on the NOMA-based system, we do not provide a
closed-form analysis for the average achievable rate of CRS-OMA.

\subsection{Outage probability for CRS-NOMA}
In this section, we derive closed-form expressions for the outage probabilities of $s_1$ and $s_2$. We define $\mathcal O_1$ as the outage event for symbol $s_1$, i.e., the event where either the relay or the destination fails to decode $s_1$ successfully. Therefore, for a given target data rate $R_1$, the outage probability for symbol $s_1$ is given by 
\begin{align}
	\Pr(\mathcal O_1) = & \Pr(C_{s_1} < R_1) = \Pr(X < \Phi) =  F_{\lambda_{sr}}(\Phi_1) + F_{\lambda_{sd}}(\Phi_1) - F_{\lambda_{sr}}(\Phi_1) F_{\lambda_{sd}}(\Phi_1), \notag
\end{align}
where $\eta_1 = 2^{2R_1} - 1$ and $\Phi_1 = \tfrac{\eta_1}{\rho(a_1 - \eta_1 a_2)}$. The system design must ensure that $a_1 > \eta_1 a_2$, otherwise the outage probability for symbol $s_1$ will always be $1$ as noted in~\cite{RelaySelectionDing}. Using~\eqref{F_lambda_tau}, the closed-form expression for the outage probability of $s_1$ in CRS-NOMA is given by
\begin{align}
	\Pr(\mathcal O_1) = \dfrac{\gamma\left( \mu_{sr}, \frac{\mu_{sr}}{\Omega_{sr}^{\alpha_{sr}}}\Phi_1^{0.5 \alpha_{sr}}\right) }{\Gamma(\mu_{sr})}+ \dfrac{\gamma\left( \mu_{sd}, \frac{\mu_{sd}}{\Omega_{sd}^{\alpha_{sd}}}\Phi_1^{0.5 \alpha_{sd}}\right)}{\Gamma(\mu_{sd})} - \dfrac{\gamma\left( \mu_{sr}, \frac{\mu_{sr}}{\Omega_{sr}^{\alpha_{sr}}}\Phi_1^{0.5 \alpha_{sr}}\right) \gamma\left( \mu_{sd}, \frac{\mu_{sd}}{\Omega_{sd}^{\alpha_{sd}}}\Phi_1^{0.5 \alpha_{sd}}\right)}{\Gamma(\mu_{sr}) \Gamma(\mu_{sd})}. \label{Out_s1}
\end{align}
Next we define the outage event for symbol $s_2$, denoted by $\mathcal O_2$, as
\begin{equation}
	\mathcal O_2 = \mathcal E_1 \cup \mathcal E_2 \cup \mathcal E_3, \notag 
\end{equation}
where $\mathcal E_1$ is the event when symbol $s_1$ cannot be successfully decoded at the relay, $\mathcal E_2$ is the event when symbol $s_1$ is successfully decoded at the relay, but symbol $s_2$ cannot be successfully decoded at the relay, and $\mathcal E_3$ is the event when both $s_1$ and $s_2$ are successfully decoded at the relay, but symbol $s_2$ cannot be successfully decoded at the destination. It is important to note that the events $\mathcal E_1, \mathcal E_2$ and $\mathcal E_3$ are disjoint. Therefore, for a given target data rate $R_2$, the outage probability for symbol $s_2$ can be given by
\begin{align}
	\Pr(\mathcal O_2) & \!= \!\!\begin{cases} \Pr(\lambda_{sr} < \Phi_{1}) + \Pr(\lambda_{sr} \geq \Phi_{1}, \lambda_{sr} < \Phi_{2}) + \Pr (\lambda_{sr} > \Phi_{2}, \lambda_{rd} < \eta_{2}/\rho); &\operatorname{if}\,\,\Phi_{1} < \Phi_{2} \\
	\Pr(\lambda_{sr} < \Phi_{1}) + \Pr(\lambda_{sr} > \Phi_{1}, \lambda_{rd} < \eta_{2}/\rho); & \operatorname{otherwise}\end{cases} \notag \\
	& =\! F_{\lambda_{sr}}\! (\Phi_{\max})\!+\! F_{\lambda_{rd}}\!\!\left(\!\!\dfrac{\eta_{2}}{\rho}\!\!\right) \!-\! F_{\lambda_{sr}}\!(\!\Phi_{\max}\!) F_{\lambda_{rd}}\!\left(\!\!\dfrac{\eta_{2}}{\rho}\!\!\right), \notag 
\end{align}
where $\eta_2 = 2^{2R_2} - 1$, $\Phi_2 = \tfrac{\eta_2}{a_2\rho}$ and $\Phi_{\max} = \max\{\Phi_1, \Phi_2\}$. Substituting the expressions for $F_{\lambda_{sr}}(\Phi_{\max})$ and $F_{\lambda_{rd}}(\eta_2/\rho)$ using~\eqref{F_lambda_tau}, the closed-form expression for the outage probability of symbol $s_2$ is given by
\begin{align}
	\Pr(\mathcal O_2) =   \dfrac{ \gamma\left( \mu_{sr}, \frac{\mu_{sr}}{\Omega_{sr}^{\alpha_{sr}}}\Phi_{\max}^{0.5 \alpha_{sr}}\right) }{\Gamma(\mu_{sr})} + \dfrac{\gamma \left(\mu_{rd}, \frac{\mu_{rd}\eta_2^{0.5 \alpha_{rd}}}{\rho^{0.5 \alpha_{rd}} \Omega_{rd}^{\alpha_{rd}}} \right)}{\Gamma(\mu_{rd})} - \dfrac{ \gamma\left( \mu_{sr}, \frac{\mu_{sr}}{\Omega_{sr}^{\alpha_{sr}}}\Phi_{\max}^{0.5 \alpha_{sr}}\right) \gamma \left(\mu_{rd}, \frac{\mu_{rd}\eta_2^{0.5 \alpha_{rd}}}{\rho^{0.5 \alpha_{rd}} \Omega_{rd}^{\alpha_{rd}}} \right)}{\Gamma(\mu_{sr}) \Gamma(\mu_{rd})}.\label{Out_s2}
\end{align}
\subsection{Diversity analysis for CRS-NOMA}
Using~\eqref{F_lambda_tau} and the series expansion of the lower incomplete Gamma function as given in~\cite[Eqn.~(8.11.4),~p.~180]{NIST}, we have 
\begin{align}
	F_{\lambda_{sr}}(\Phi_1) = & \dfrac{\mu_{sr}^{\mu_{sr}} \Phi_1^{0.5 \alpha_{sr} \mu_{sr}}}{\Gamma(\mu_{sr}) \Omega_{sr}^{\alpha_{sr} \mu_{sr}}} \exp \left( \dfrac{-\mu_{sr} \Phi_1^{0.5 \alpha_{sr}}}{\Omega_{sr}^{\alpha_{sr}}}\right) \sum_{k = 0}^{\infty} \dfrac{\mu_{sr}^{k} \Phi_1^{0.5 k\alpha_{sr}}}{\Omega_{sr}^{k\alpha_{sr}} (\mu_{sr})_{k+1}}, \notag
\end{align}
where $(x)_y = \Gamma(x + y)/\Gamma(x)$ is the Pochhammer symbol. Using the series expansion of the exponential function, we have 
\begin{align}
	F_{\lambda_{sr}}(\Phi_1) = \dfrac{\mu_{sr}^{\mu_{sr}}}{\Omega_{sr}^{\alpha_{sr}\mu_{sr}}} \sum_{l = 0}^{\infty} \sum_{k = 0}^{\infty} \dfrac{(-1)^l \Phi_1^{0.5 \alpha_{sr}(\mu_{sr} + l + k)} \mu_{sr}^{l + k}}{l! \, \Omega_{sr}^{\alpha_{sr} (l + k)} \Gamma(\mu_{sr} + k + 1)}. \notag 
\end{align}
Replacing $\Phi_1$ by $\tfrac{\eta_1}{\rho (a_1 - \eta_1 a_2)}$, the asymptotic behavior ($\rho \to \infty$) of $F_{\lambda_{sr}}(\Phi_1)$ is given by 
\begin{align}
	F_{\lambda_{sr}}(\Phi_1) = & \dfrac{\mu_{sr}^{\mu_{sr}} \eta_1^{0.5 \alpha_{sr} \mu_{sr}} \rho^{-0.5 \alpha_{sr}\mu_{sr}}}{\Omega_{sr}^{\alpha_{sr}\mu_{sr}}(a_1 - \eta_1 a_2)^{0.5 \alpha_{sr} \mu_{sr}} \Gamma(\mu_{sr} + 1)}+ \mathscr O \left( \rho^{-(0.5 \alpha_{sr} \mu_{sr} + 1)}\right), \label{F_lambdaSR_asymp}
\end{align}
where $\mathscr O(\cdot)$ is the Landau symbol. It is clear from~\eqref{F_lambdaSR_asymp} that for high transmit SNR $\rho$, $F_{\lambda_{sr}}(\Phi_1)$ decays as $\rho^{-0.5 \alpha_{sr} \mu_{sr}}$. Analogously, it can be shown that at high transmit SNR, $F_{\lambda_{sd}}(\Phi_1)$ decays as $\rho^{-0.5 \alpha_{sd} \mu_{sd}}$ and $F_{\lambda_{sr}}(\Phi_1) F_{\lambda_{sd}}(\Phi_1)$ decays as $\rho^{-0.5 (\alpha_{sr} \mu_{sr} + \alpha_{sd} \mu_{sd})}$. Therefore, using~\eqref{Out_s1}, it is straightforward to conclude that the diversity order of symbol $s_1$ is $0.5 \min\{\alpha_{sr} \mu_{sr}, \alpha_{sd} \mu_{sd}\}$. Using similar arguments, it can be deduced from~\eqref{Out_s2} that the diversity order of symbol $s_2$ is $0.5 \min\{\alpha_{sr} \mu_{sr}, \alpha_{rd} \mu_{rd}\}$.
\section{Results and Discussion}
\begin{table}[t]
\centering
\caption{Optimal value of $a_2$ for different transmit SNR for $\Omega_{sr} = \Omega_{rd} = 10$, $\Omega_{sd} = 1$ and $R_1 = 1$ bps/Hz.}
\begin{tabular}{|c|c|c|c|c|c|c|c|c|}
\hline
\multirow{4}{*}{$(\alpha, \mu)$} & \multicolumn{8}{c|}{Transmit SNR $(\rho)$ in dB}    \\ \cline{2-9} 
                                 & 0   & 5   & 10   & 15   & 20   & 25   & 30   & 35   \\ \cline{2-9} 
                                 & \multicolumn{8}{c|}{Optimal value of $a_2$}                          \\ \hline
(2, 1)                           & 0.24 & 0.24 & 0.24 & 0.24 & 0.17 & 0.1  & 0.06 & 0.04 \\ \hline
(2, 2)                           & 0.24 & 0.24 & 0.24 & 0.24 & 0.21 & 0.14 & 0.09 & 0.06 \\ \hline
(3, 1)                           & 0.24 & 0.24 & 0.24 & 0.24 & 0.21 & 0.14 & 0.09 & 0.06 \\ \hline
(4, 1)                           & 0.24 & 0.24 & 0.24 & 0.24 & 0.24 & 0.17 & 0.12 & 0.08 \\ \hline
\end{tabular}
\label{Opt_table}
\end{table}
In this section, we present analytical and numerical\footnote{We do not realize the actual scenario for numerical computation, but rather generate the random variables and then evaluate~\eqref{C_s1_int},~\eqref{C_s2_int} and~\eqref{C_OMA}.} results for the average achievable rate, outage probability and diversity order of the NOMA-based cooperative relaying system. Table~\ref{Opt_table} shows the optimal value of the power allocation coefficient $a_2$, which maximizes the average achievable rate for CRS-NOMA, for different values of $\alpha_{sr} = \alpha_{sd} = \alpha_{rd} = \alpha$, $\mu_{sr} = \mu_{sd} = \mu_{rd} = \mu$, and $\rho$. The necessary constraint $a_1 > \eta_1 a_2$ (as noted in~Section~III-B) implies a possible range $0 < a_2 < 2^{-2R_1}$. Therefore, the optimization of $a_2$ was performed using an exhaustive search over the $M$-element discrete set $a_2 \in \{ \epsilon, 2 \epsilon, 3 \epsilon, \ldots, M \epsilon \}$, where $M$ is a positive integer and $\epsilon = 2^{-2R_1}/ (M+1)$. In Table I, we consider $M = 24 \ (\epsilon = 0.01)$, $R_1 = R_2 =1$ bps/Hz, $\Omega_{sr} = \Omega_{rd} = 10$ and $\Omega_{sd} = 1$.
Note that the average achievable rate for CRS-NOMA depends on the target data rate for symbol $s_1$. The expressions derived in this paper are valid also for the case when $\alpha_{sr} \neq \alpha_{sd} \neq \alpha_{rd}$ and $\mu_{sr} \neq \mu_{sd} \neq \mu_{rd}$.
 \begin{figure}[t]
\centering
\includegraphics[width=0.7\linewidth]{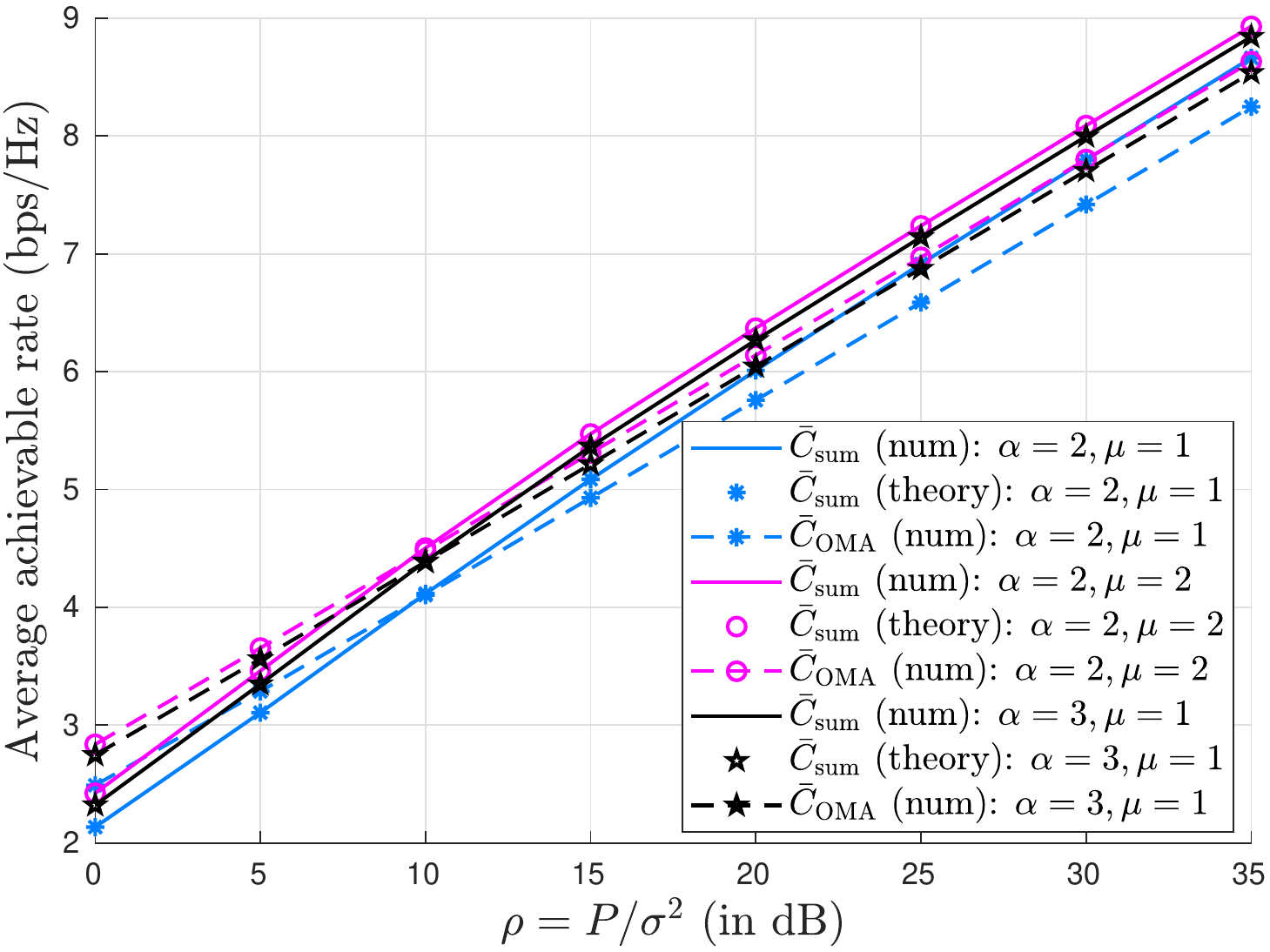}
\caption{Average achievable rate (with optimal power allocation in the NOMA case) for CRS with $\Omega_{sr}~=~\Omega_{rd}~=~10$, $\Omega_{sd} = 1$ and $R_1 = 1$ bps/Hz.}
\label{Cap_CRS}
\end{figure}
\begin{figure}[t]
\centering
\includegraphics[width=0.7\linewidth]{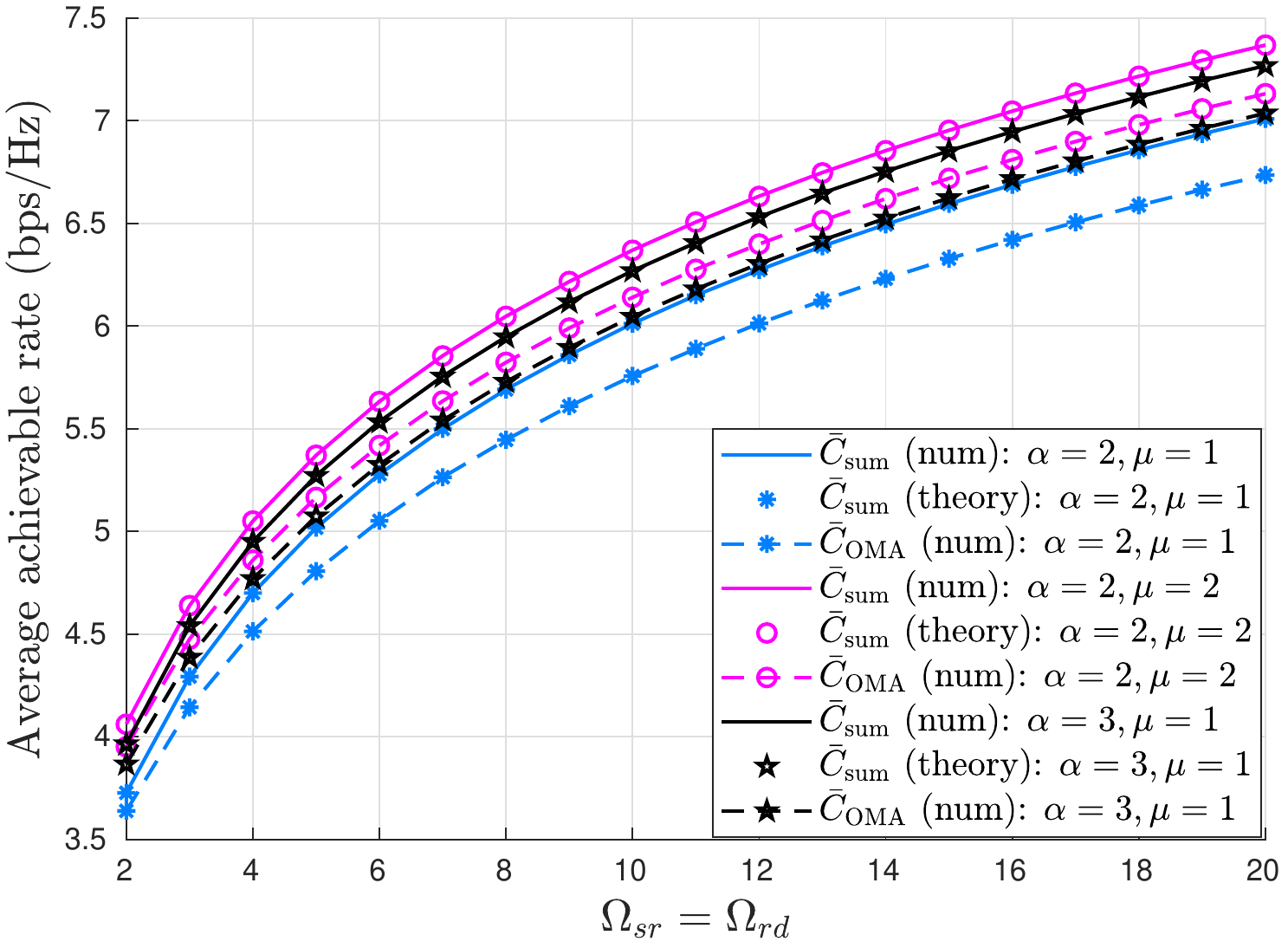}
\caption{Average achievable rate (with optimal power allocation in  the NOMA case) for CRS with transmit SNR $\rho = 20$ dB, $\Omega_{sd} = 1$ and $R_1 = 1$ bps/Hz.}
\label{Cap_Omega}
\end{figure}

Fig.~\ref{Cap_CRS} shows the average achievable rate for the CRS-NOMA (with optimal $a_2$ for each $\rho$) and CRS-OMA for $\alpha = 2, \mu = 1$ (i.e., Rayleigh fading); $\alpha = 2, \mu = 2$ (i.e., Nakagami-$m$ fading with $m = 2$) and $\alpha = 3, \mu = 1$ (i.e., Weibull fading with shape parameter $\alpha = 3$). The figure shows a perfect agreement between the numerical and analytical plots for the NOMA system. It is also clear from the figure that for high transmit SNR $\rho$, the NOMA-based CRS outperforms its OMA-based counterpart in terms of spectral efficiency. 

For the same set of channel parameters, Fig.~\ref{Cap_Omega} shows the average achievable rate for CRS-NOMA (with optimal $a_2$ for each $\Omega_{sr} = \Omega_{rd}$) and CRS-OMA at a fixed transmit SNR $\rho = 20$ dB and $\Omega_{sd} = 1$ against $\Omega_{sr} = \Omega_{rd}$. It is clear from the figure that the CRS-NOMA outperforms CRS-OMA by achieving significant gain in spectral efficiency.

Fig.~\ref{Cap_alpha} shows the average achievable rate of CRS-NOMA (with optimal $a_2$ for each $\rho$) and CRS-OMA for a fixed value of the clustering parameter $\mu = 1$ and different values of the nonlinearity parameter $\alpha$. It is clear from the figure that as $\alpha$ increases, the spectral efficiency of the CRS also increases. Fig.~\ref{Cap_mu} shows the average achievable rate of CRS-NOMA (with optimal $a_2$ for each $\rho$) and CRS-OMA for a fixed value of the nonlinearity parameter $\alpha = 2$ and different values of the clustering parameter $\mu$. It is evident from the figure that with an increase in the value of the clustering parameter, the spectral efficiency of the CRS increases. It can also be noted from Figs.~\ref{Cap_alpha} and~\ref{Cap_mu} that with increase in the value of nonlinearity/clustering parameter, the SNR at which the CRS-NOMA outperforms its OMA counterpart becomes higher.

Figs.~\ref{Out_s1_fig} and~\ref{Out_s2_fig} show the outage probabilities of symbols $s_1$ and $s_2$, respectively, in CRS-NOMA, evaluated using the methodology of~ Section~III-B for $a_2 = 0.1$, $\Omega_{sr} = \Omega_{sd} = 10$ and $\Omega_{sd} = 1$. It is evident that both outage probabilities depend on the fading parameters $\alpha$ and $\mu$, and there is a consistent agreement between the observed diversity order for each symbol and the analytical diversity order results derived in~Section~III-C.
\section{Conclusion}
In this work, an exact closed-form expression for the average achievable rate of a CRS-NOMA in generalized $\alpha$-$\mu$ fading was derived and validated through numerical experiments. The derived expressions can be used to analyze the performance of the CRS-NOMA over diverse small-scale fading channels including Rayleigh, Nakagami-$m$, exponential, Gamma and Weibull, as well as some large-scale fading channels. Our results show that the optimal power allocation for maximizing the achievable rate of CRS-NOMA depends on the target data rate requirement. Our results also demonstrate the effect of the nonlinearity and the clustering parameters on the average achievable rate of the CRS. Outage probability and diversity order were also analyzed, and it was shown that a simple relationship exists between the diversity order for each symbol, and the nonlinearity and clustering parameters of the fading links.
\section*{Acknowledgment}
This publication has emanated from research conducted with the financial support of Science Foundation Ireland (SFI) and is co-funded under the European Regional Development Fund under Grant Number 13/RC/2077.
\begin{figure}[t]
\centering
\includegraphics[width=0.7\linewidth]{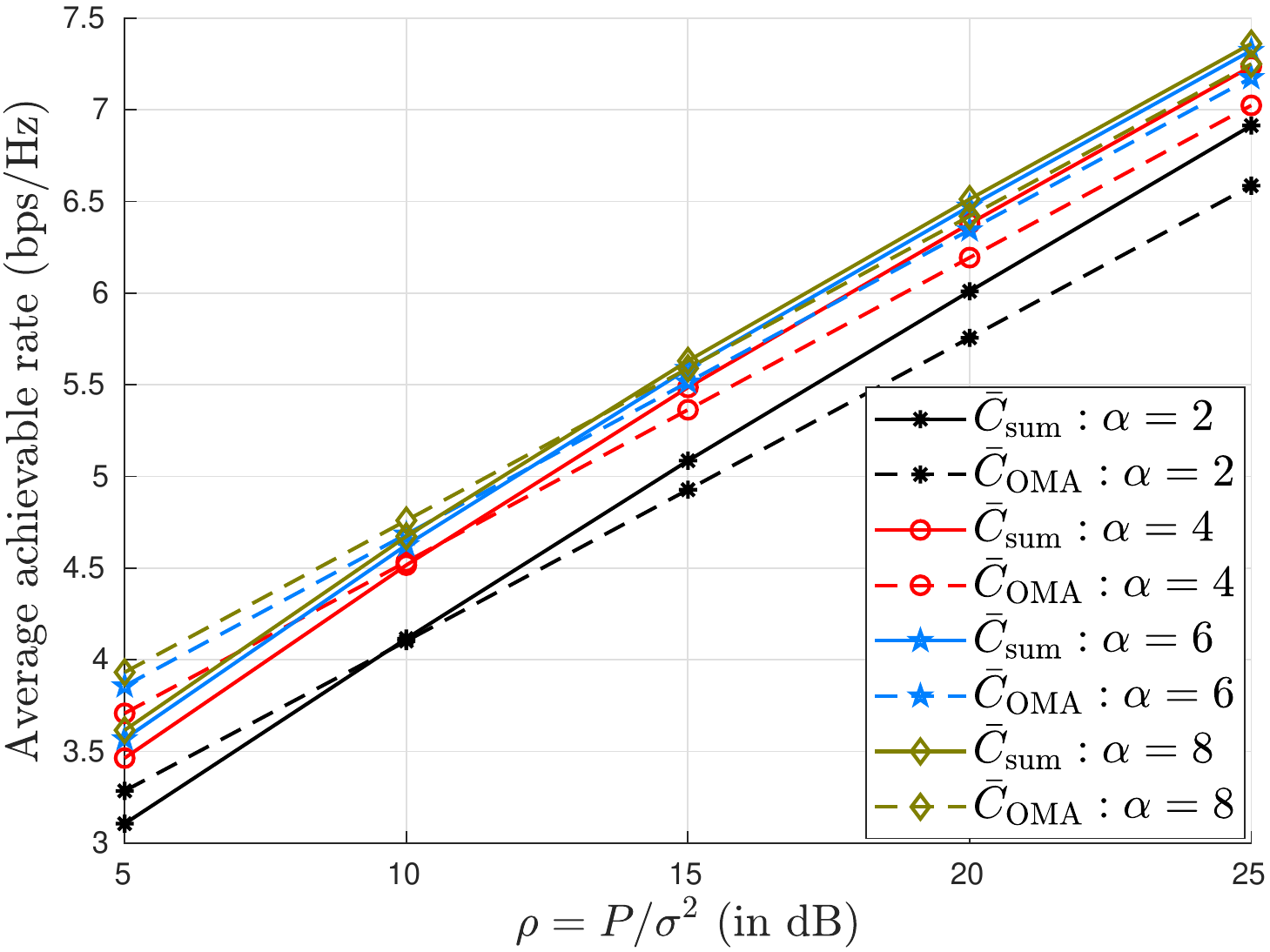}
\caption{Average achievable rate (with optimal power allocation in the NOMA case) for CRS with $\Omega_{sr} = \Omega_{rd} = 10$, $\Omega_{sd} = 1$, $R_1 = 1$ bps/Hz and $\mu = 1$.}
\label{Cap_alpha}
\end{figure}
\begin{figure}[t]
\centering
\includegraphics[width=0.7\linewidth]{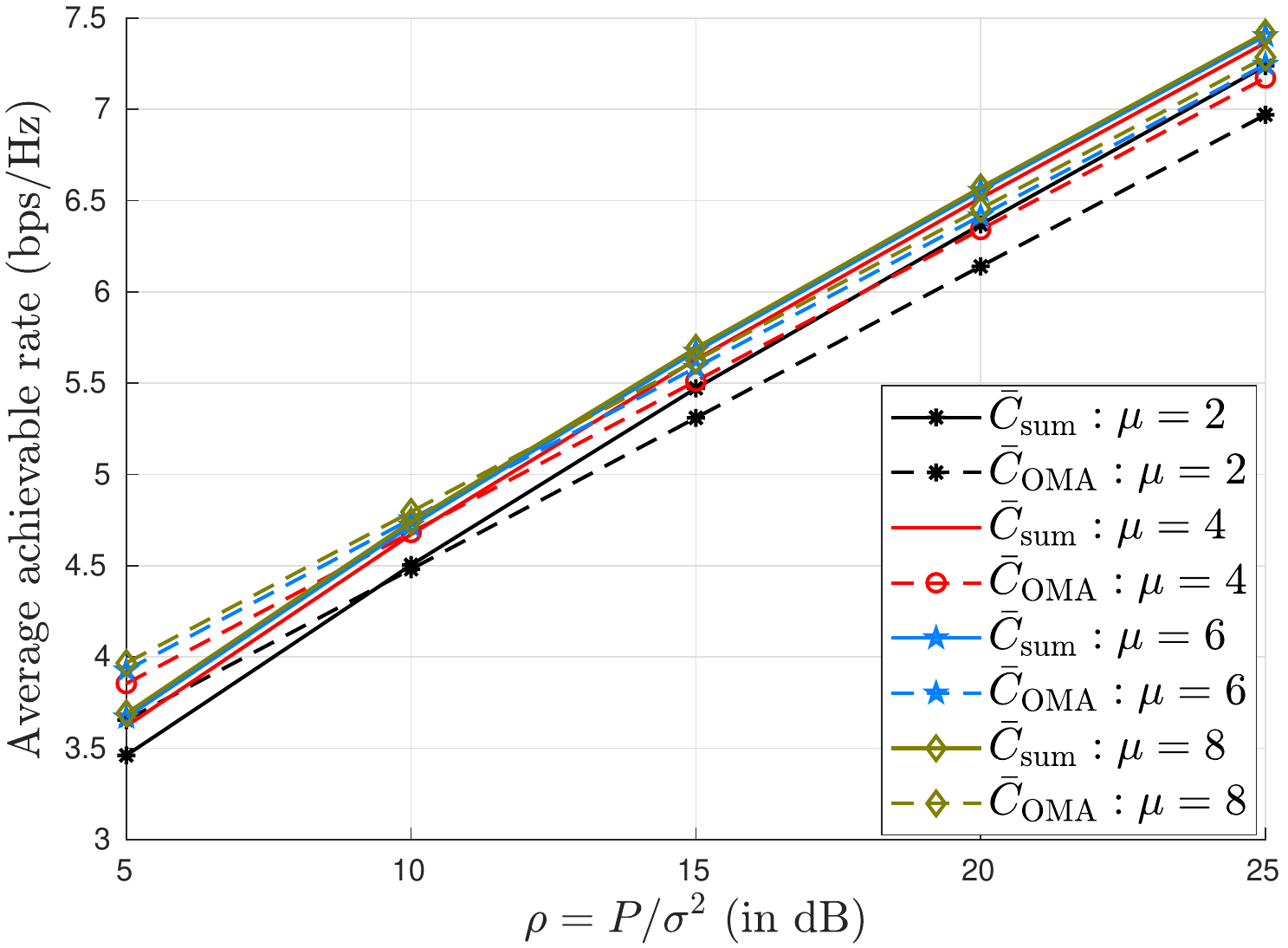}
\caption{Average achievable rate (with optimal power allocation in the NOMA case) for CRS with $\Omega_{sr} = \Omega_{rd} = 10$, $\Omega_{sd}\! = \!1$, $R_1 = 1$ bps/Hz  and $\alpha = 2$.}
\label{Cap_mu}
\end{figure}
\appendices
\section{Proof of Theorem~1}
The PDF of $X = \min\{\lambda_{sr}, \lambda_{sd}\}$ can be given by
\begin{align}
	\!f_X(x) \!=\! f_{\lambda_{sr}}(x)\!-\! f_{\lambda_{sr}}(x) F_{\lambda_{sd}}(x)\! +\! f_{\lambda_{sd}}(x)\! -\! f_{\lambda_{sd}}(x) F_{\lambda_{sr}}\!(x). \label{fX}
\end{align}
The expressions for $f_{\lambda_{sr}}(x)$, $f_{\lambda_{sd}}(x)$,$F_{\lambda_{sr}}(x)$ and $F_{\lambda_{sd}}(x)$ can be obtained using~\eqref{F_lambda_tau} and~\eqref{f_lambda_tau}. Substituting the expression of $f_X(x)$ from~\eqref{fX} into~\eqref{C_s1_int}, we have 
\begin{align}
	& \mathbb I_1 \!\!=\!\!\! \underbrace{\int_{0}^{\infty}  \!\!\!\!\!\! \ln(1 + \rho x) f_{\lambda_{sr}}(x) dx}_{I_1} \! -\!\!\! \underbrace{\int_{0}^{\infty}  \!\!\!\!\!\!\ln(1 + \rho x) f_{\lambda_{sr}}(x) F_{\lambda_{sd}}(x) dx}_{I_2}
	+  \underbrace{\int_{0}^{\infty} \!\!\!\! \ln(1 + \rho x) f_{\lambda_{sd}}(x) dx}_{I_3} \!-\! \underbrace{\int_{0}^{\infty}\!\!\!\! \ln(1 + \rho x) f_{\lambda_{sd}}(x) F_{\lambda_{sr}}(x) dx}_{I_4} \notag
\end{align}
Solving for $I_1$, we have 
\begin{align}
	I_1  =  &  \int_{0}^{\infty} \ln(1 + \rho x)f_{\lambda_{sr}}(x)dx 
	=  \dfrac{0.5\alpha_{sr}\mu_{sr}^{\mu_{sr}}}{\Omega_{sr}^{\alpha_{sr}\mu_{sr}} \Gamma(\mu_{sr})} \!\!\int_0^{\infty}\!\!\!\!\!\!\!\ln(\!1 \!+\! \rho x\!) \exp \!\!\left(\! \! \dfrac{-\mu_{sr} x^{\tfrac{\alpha_{sr}}{2}}}{\Omega_{sr}^{\alpha_{sr}}}\!\!\right)\!x^{\tfrac{\alpha_{sr} \mu_{sr}}{2} - 1}dx. \notag
\end{align}
Representing the logarithmic and the exponential functions in terms of Meijer's G-function using~\cite[Eqn.~(11)]{Reduce} yields
\begin{align}
	\!\!\!\!I_1 \! = & \dfrac{0.5\alpha_{sr}\mu_{sr}^{\mu_{sr}}}{\Omega_{sr}^{\alpha_{sr}\mu_{sr}} \Gamma(\mu_{sr})} \!\!\int_0^{\infty} \!\!\!\!x^{0.5\alpha_{sr} \mu_{sr} - 1} G_{2, 2}^{1, 2}\left( \rho x\left\vert \begin{smallmatrix} 1, & 1 \\[0.6em] 1, & 0 \end{smallmatrix}\right.\right)  G_{0, 1}^{1, 0} \left( \dfrac{\mu_{sr} x^{0.5\alpha_{sr}}}{\Omega_{sr}^{\alpha_{sr}}} \left\vert \begin{smallmatrix}\text{--} \\[0.6em] 0 \end{smallmatrix} \right.\right) dx. \notag
\end{align}
\begin{figure}[t]
\centering
\includegraphics[width=0.7\linewidth]{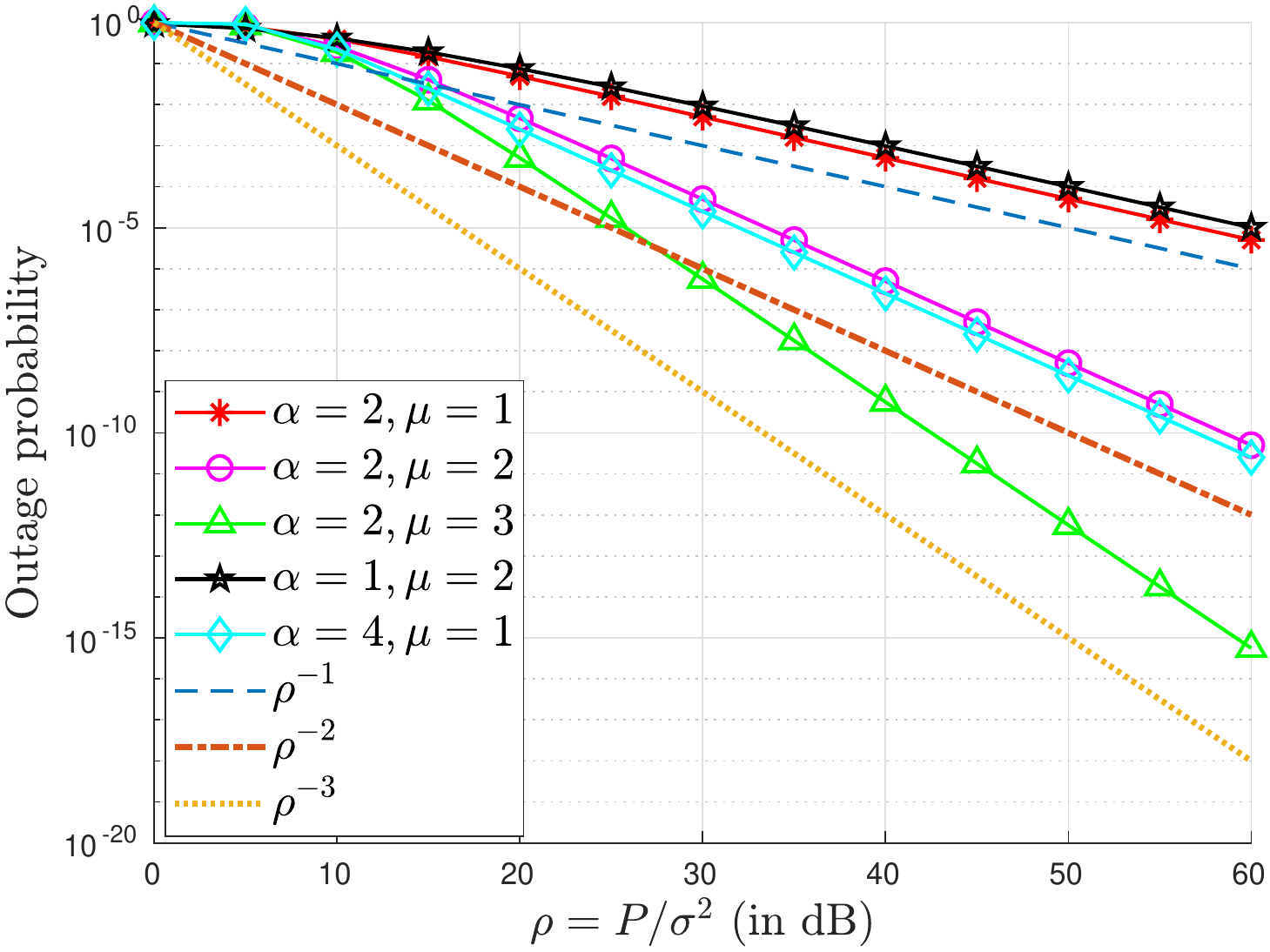}
\caption{Outage probability for CRS-NOMA for symbol $s_1$ for $a_2 = 0.1$, $\Omega_{sr} = \Omega_{sd} = 10$, $\Omega_{rd} = 1$ and $R_1 = 1$ bps/Hz.}
\label{Out_s1_fig}
\end{figure}
\begin{figure}[t]
\centering
\includegraphics[width=0.7\linewidth]{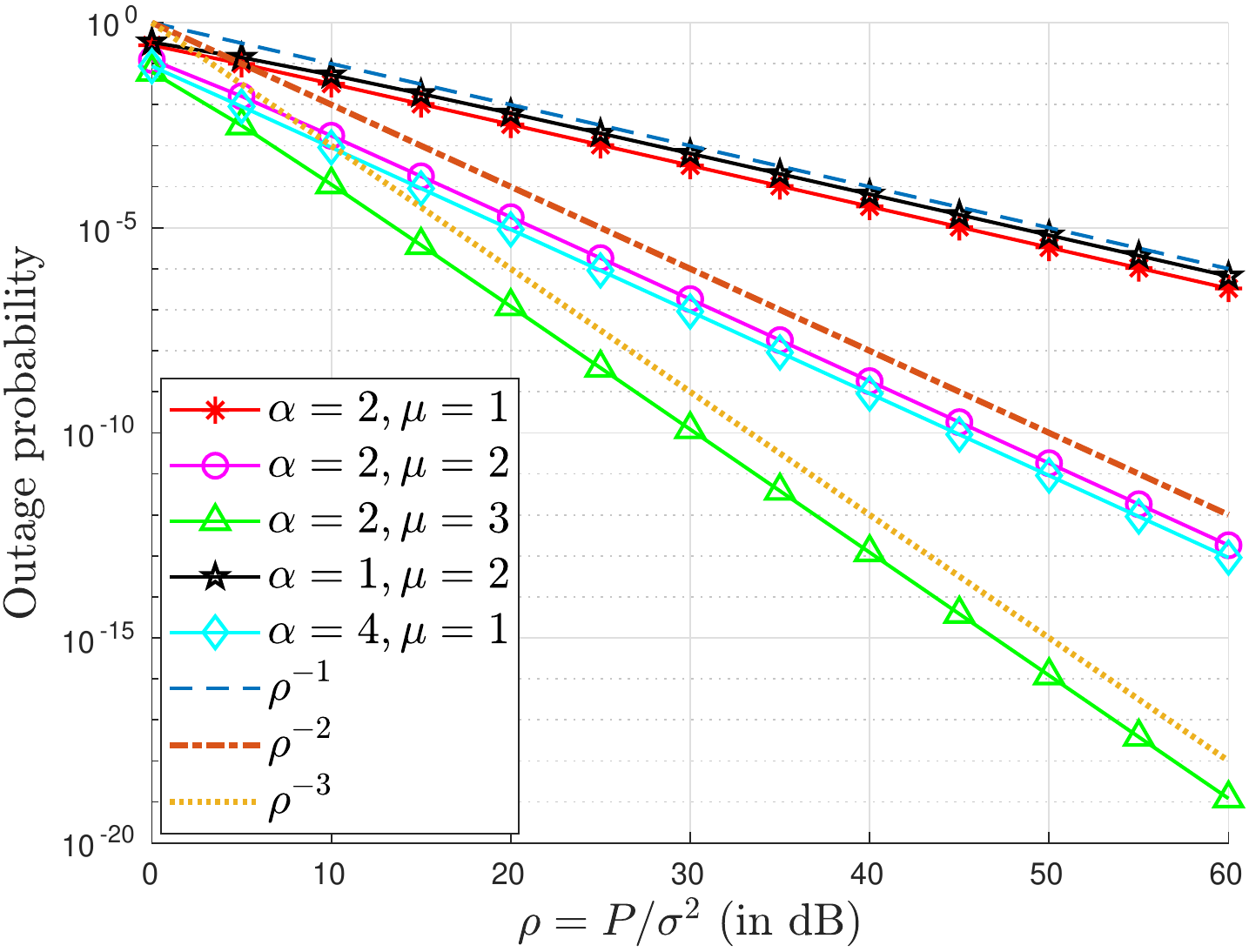}
\caption{Outage probability for CRS-NOMA for symbol $s_2$ for $a_2 = 0.1$, $\Omega_{sr} = \Omega_{sd} = 10$, $\Omega_{rd} = 1$ and $R_1 = R_2 = 1$ bps/Hz..}
\label{Out_s2_fig}
\end{figure}\\
Solving the integration above using~\cite[Eqns.~(7),~(21)]{Reduce}, $I_1$ reduces to~\eqref{I1}. Analogously, the closed-form expression for $I_3$ is given by~\eqref{I3}. Similarly, 
\begin{align}
	I_2 \!= & \int_{0}^{\infty}\ln(1 + \rho x)f_{\lambda_{sr}}(x) F_{\lambda_{sd}}(x)dx \notag \\
	= & \dfrac{0.5\alpha_{sr}\mu_{sr}^{\mu_{sr}}}{\Omega_{sr}^{\alpha_{sr}\mu_{sr}} \Gamma(\mu_{sr}) \Gamma(\mu_{sd})} \!\!\int_0^{\infty}\!\!\!\!\ln(\!1 \!+\! \rho x\!) \exp \left(\! \! \dfrac{-\mu_{sr} x^{0.5\alpha_{sr}}}{\Omega_{sr}^{\alpha_{sr}}}\!\!\right)  \gamma \left( \mu_{sd}, \dfrac{\mu_{sd} x^{0.5\alpha_{sd}}}{\Omega_{sd}^{\alpha_{sd}}}\right) x^{0.5\alpha_{sr} \mu_{sr} - 1} dx. \notag
\end{align}
Representing the logarithmic and the exponential function in terms of Meijer's G-function using~\cite[Eqn.~(11)]{Reduce} and the lower-incomplete Gamma function in terms of Meijer's G-function using~\cite[Eqn.~(2)]{Secrecy} yields
\begin{align}
	I_2 = &  \dfrac{0.5\alpha_{sr}\mu_{sr}^{\mu_{sr}}}{\Omega_{sr}^{\alpha_{sr}\mu_{sr}} \Gamma(\mu_{sr}) \Gamma(\mu_{sd})} \!\!\int_0^{\infty} \!\!\!\!\!x^{\tfrac{\alpha_{sr} \mu_{sr}}{2} - 1} G_{2, 2}^{1, 2}\left(\! \rho x \left\vert \begin{smallmatrix}1, & 1 \\[0.6em]1, & 0 \end{smallmatrix}\right.\!\!\right)  G_{0, 1}^{1, 0}\!\left( \dfrac{\mu_{sr}x^{\tfrac{\alpha_{sr}}{2}}}{\Omega_{sr}^{\alpha_{sr}}} \left\vert \begin{smallmatrix} \text{--} \\[0.6em] 0\end{smallmatrix}\right.\right) \!G_{1, 2}^{1, 1}\!\left(\!\! \dfrac{\mu_{sd}x^{\tfrac{\alpha_{sd}}{2}}}{\Omega_{sd}^{\alpha_{sd}}} \left\vert \begin{smallmatrix} 1 \\[0.6em] \mu_{sd}, \, 0\end{smallmatrix}\right.\!\right) dx. \notag
\end{align}
Using the relation between Meijer's G-function and Fox H-function~(c.f.~\cite[Eqn.~(6.2.8)]{Springer})
\begin{align}
	G_{p, q}^{m, n}\left[ x \left\vert \begin{smallmatrix} \upsilon_1, \ldots, \upsilon_p \\[0.6em] \nu_1, \ldots, \nu_q\end{smallmatrix}\right.\right] = H_{p, q}^{m, n}\left[ x \left\vert \begin{smallmatrix} (\upsilon_1,1), \ldots, (\upsilon_p,1) \\[0.6em] (\nu_1,1),  \ldots, (\nu_q, 1)\end{smallmatrix}\right.\right], \notag
\end{align}
and substituting $\tilde x = x^{0.5 \alpha_{sr}}$ yields
\begin{align}
	I_2 =  & \dfrac{\mu_{sr}^{\mu_{sr}} \Omega_{sr}^{-\alpha_{sr}\mu_{sr}}}{ \Gamma(\mu_{sr}) \Gamma(\mu_{sd})} \!\!\int_0^{\infty} \!\!\!\!\!\tilde x^{\mu_{sr} - 1} H_{2, 2}^{1, 2}\left(\!\! \rho \tilde x^{2/\alpha_{sr}} \left\vert \begin{smallmatrix}(1, 1), & (1, 1) \\[0.6em](1, 1), & (0, 1) \end{smallmatrix}\right.\!\!\right)  \!H_{0, 1}^{1, 0}\!\left(\!\! \dfrac{\mu_{sr}\tilde x}{\Omega_{sr}^{\alpha_{sr}}} \left\vert \begin{smallmatrix} \text{--} \\[0.6em] (0, 1)\end{smallmatrix}\right.\!\!\!\right) \!H_{1, 2}^{1, 1}\!\left(\!\! \dfrac{\mu_{sd}\tilde x^{\tfrac{\alpha_{sd}}{\alpha_{sr}}}}{\Omega_{sd}^{\alpha_{sd}}} \left\vert \begin{smallmatrix} (1, 1) \\[0.6em] (\mu_{sd}, 1), (0, 1)\end{smallmatrix}\right.\!\!\!\right) d\tilde x, \notag
\end{align}
Solving the integration above using~\cite[Eqn.~(2.3)]{Mittal}, the closed-form expression for $I_2$ reduces to~\eqref{I2}. Solving in a similar fashion, the closed-form expression for $I_4$ can be given by~\eqref{I4}.

Also, using~\eqref{fX} and~\eqref{C_s1_int}, we have 
\begin{align}
	& \!\!\!\mathbb I_2 \!\!= \!\!\underbrace{\int_{0}^{\infty} \!\!\!\!\!\!\! \ln(1\! + \!a_2\rho x) f_{\lambda_{sr}}\!(x) dx}_{I_5} \!-\!\!\!\underbrace{\int_{0}^{\infty} \!\!\!\!\!\!\!\ln(1 \!+\! a_2\rho x) f_{\lambda_{sr}}\!(x) F_{\lambda_{sd}}\!(x) dx}_{I_6} +  \!\underbrace{\int_{0}^{\infty} \!\!\!\!\! \ln(1 \!+\! a_2\rho x) f_{\lambda_{sd}}\!(x) dx}_{I_7} \!-\!\! \underbrace{\int_{0}^{\infty}\!\!\!\!\! \ln(1 \!+\! a_2 \rho x) f_{\lambda_{sd}}\!(x) F_{\lambda_{sr}}\!(x) dx}_{I_8} \notag
\end{align}
The closed-form expressions for $I_5, I_6, I_7$ and $I_8$ can be obtained by replacing $\rho$ by $a_2 \rho$ in~\eqref{I1}-\eqref{I4}, respectively. Substituting the expressions for $I_1, I_2, \ldots, I_8$ into~\eqref{C_s1_theorem} yields the closed-form expression for the average achievable rate for symbol $s_1$ of CRS-NOMA in $\alpha$-$\mu$ fading.
\section{Proof of Theorem~2}
Using~\eqref{F_lambda_tau} and a transformation of random variables, the CDF of $a_2 \lambda_{sr}$ can be expressed as 
\begin{align}
	F_{a_2 \lambda_{sr}}(y) = F_{\lambda_{sr}}\left( \dfrac{y}{a_2}\right) =  \dfrac{1}{\Gamma(\mu_{sr})}\gamma \left( \mu_{sr}, \dfrac{\mu_{sr}y^{0.5 \alpha_{sr}}}{a_2^{0.5 \alpha_{sr}} \Omega_{sr}^{\alpha_{sr}}}\right). \notag
\end{align}\balance
Similarly, using~\eqref{f_lambda_tau} and a transformation of random variables, the PDF of $a_2 \lambda_{sr}$ can be expressed as
\begin{align}
	f_{a_2 \lambda_{sr}}(y) = & \dfrac{1}{a_2} f_{\lambda_{sr}} \left(\dfrac{y}{a_2} \right) =   \dfrac{\alpha_{sr} \mu_{sr}^{\mu_{sr}} y^{0.5\alpha_{sr} \mu_{sr} - 1}}{2a_2^{0.5\alpha_{sr} \mu_{sr}} \Omega_{sr}^{\alpha_{sr} \mu_{sr}} \Gamma(\mu_{sr})} \exp \left( \dfrac{-\mu_{sr} y^{0.5\alpha_{sr}}}{a_2^{0.5\alpha_{sr}}\Omega_{sr}^{\alpha_{sr}}}\right).\notag
\end{align}
Therefore, the PDF of $Y$ can be given by
\begin{align}
	f_Y(x) = & f_{a_2 \lambda_{sr}}(y) - f_{a_2 \lambda_{sr}}(y)F_{\lambda_{rd}}(y) + f_{\lambda_{rd}}(y) - f_{\lambda_{rd}}(y)F_{a_2 \lambda_{sr}}(y). \notag 
\end{align}
Substituting the expression of $f_Y(y)$ shown above into~\eqref{C_s2_int}, we have 
\begin{equation}
		I_9 \triangleq \int_{0}^{\infty}\ln(1 + \rho y)f_{a_2 \lambda_{sr}}(y)dy, \quad I_{10} \triangleq \int_{0}^{\infty}\ln(1 + \rho y)f_{a_2 \lambda_{sr}}(y) F_{\lambda_{rd}}(y)dy, \notag
\end{equation}
\begin{equation}
	I_{11} \triangleq \int_{0}^{\infty}\ln(1 + \rho y)f_{\lambda_{rd}}(y)dy, \quad
	I_{12} \triangleq \int_{0}^{\infty}\ln(1 + \rho y)f_{a_2 \lambda_{rd}}(y) F_{a_2\lambda_{sr}}(y)dy. \notag
\end{equation}		
Solving the integrals shown above using the method in Appendix~A, closed-form expressions for $I_9, I_{10}, I_{11}$ and $I_{12}$ can be shown to be given by~\eqref{I9}-\eqref{I12}. 
\bibliographystyle{IEEEtran}
\bibliography{ICC}

\end{document}